\documentclass[a4paper,12pt]{article}
\usepackage[utf8]{inputenc}
\usepackage{amsthm}
\usepackage{amssymb}
\usepackage{amsmath}
\usepackage{epsfig}
\usepackage{colonequals}
\usepackage{enumerate}
\usepackage{hyperref}

\newcommand{\Aa}{{\mathcal A}}
\newcommand{\GG}{{\mathcal G}}

\newcommand{\KK}{{\mathcal K}}

\newcommand{\RR}{{\mathcal R}}

\newcommand{\FF}{{\mathcal F}}
\newcommand{\CC}{{\mathcal C}}
\newcommand{\PP}{{\mathcal P}}
\newcommand{\TT}{{\mathcal T}}

\newcommand{\tw}{\text{tw}}

\newtheorem{theorem}{Theorem}[section]

\newtheorem{corollary}[theorem]{Corollary}
\newtheorem{lemma}[theorem]{Lemma}

\begin{document}
\title{Baker game and polynomial-time approximation schemes}
\author{%
     Zden\v{e}k Dvo\v{r}\'ak\thanks{Computer Science Institute (CSI) of Charles University,
           Malostransk{\'e} n{\'a}m{\v e}st{\'\i} 25, 118 00 Prague, 
           Czech Republic. E-mail: \protect\href{mailto:rakdver@iuuk.mff.cuni.cz}{\protect\nolinkurl{rakdver@iuuk.mff.cuni.cz}}.}
}
\date{}
\maketitle
\newcommand{\seq}{\mathbf}
\newcommand{\tail}{\text{tail}}
\newcommand{\head}{\text{head}}

\begin{abstract}
Baker~\cite{baker1994approximation} devised a technique to obtain approximation schemes for many
optimization problems restricted to planar graphs; her technique was later extended 
to more general graph classes.  In particular, using the Baker's technique and the minor structure theorem,
Dawar et al.~\cite{dawar2006approximation} gave Polynomial-Time Approximation Schemes (PTAS) for all monotone
optimization problems expressible in the first-order logic when restricted to a proper minor-closed class of graphs.
We define a \emph{Baker game} formalizing the notion of repeated application of Baker's technique interspersed with vertex
removal, prove that monotone optimization problems expressible in the first-order logic
admit PTAS when restricted to graph classes in which the Baker game
can be won in a constant number of rounds, and prove \emph{without} use of the minor structure theorem that
all proper minor-closed classes of graphs have this property.
\end{abstract}

\section{Introduction}

Baker~\cite{baker1994approximation} devised to polynomial-time approximation schemes for
a range of problems (including maximum independent set, minimum dominating set, largest $H$-matching, minimum vertex cover,
and many others) when restricted to planar graphs.  Her technique (formulated in modern terms) is based on the fact that
if $(V_0,\ldots, V_d)$ is a partition of vertices of a connected planar graph $G$ according to their distance from
an arbitrarily chosen vertex, then for any positive integer $s$ and $0\le i\le d-s+1$, the treewidth of
$G[V_i\cup V_{i+1}\cup \ldots\cup V_{i+s-1}]$ is bounded by a function of $s$ (more specifically, it is at most $3s$~\cite{rs3}),
and hence many natural problems can be exactly solved for these subgraphs in linear time~\cite{courcelle}.

It is natural to ask whether these algorithms can be extended to larger classes $\GG$ of graphs.
A \emph{layering} of a graph $G$ is a function $\lambda:V(G)\to \mathbb{Z}$ such that $|\lambda(u)-\lambda(v)|\le 1$ for every edge $uv\in E(G)$;
one should visualize the vertices of $G$ partitioned into layers $\lambda^{-1}(i)$ for $i\in \mathbb{Z}$, with edges of $G$ only allowed inside
the layers and between the consecutive layers.
Let us say that a class $\GG$ has \emph{bounded treewidth layerings} if for some function $f$,
each graph $G\in\GG$ has a layering $\lambda$ such that $G[\lambda^{-1}(I)]$ has treewidth at most $f(|I|)$
for any finite interval $I$ of consecutive integers.
Baker's technique directly extends to all such graph classes (assuming that a suitable layering can be found
in polynomial time).  A natural obstruction for the existence of bounded treewidth layerings is as follows:
let $G_n$ be the graph obtained from the $n\times n$ grid by adding a universal vertex adjacent to all other vertices.
Then each layering of $G_n$ has at most three non-empty layers, and $\tw(G_n)>n$.  Hence, if a class $\GG$ has bounded treewidth layerings,
then it contains only finitely many of the graphs $\{G_n:n\ge 1\}$.  Conversely, Eppstein~\cite{eppstein00} proved that
minor-closed classes that do not contain all such graphs $G_n$ have bounded treewidth layerings.

As another obstruction, let $U_n$ denote the graph obtained from an $n\times n\times n$ grid by adding all diagonals to its unit
subcubes.  Although the graphs $U_n$ have bounded maximum degree and very simple structure, 
Dvo\v{r}\'ak et al.~\cite{gridtw} proved that for every integer $k$, there exists $n_0$ such that for all $n\ge n_0$,
the vertex set of $U_n$ cannot be partitioned into two parts both inducing subgraphs of treewidth at most $k$.
This prevents existence of bounded treewidth layerings of $U_n$, since otherwise the partition of the graph into odd and even numbered layers
would give a contradiction.

Of course, one can work around these obstructions:
\begin{itemize}
\item Each of the graphs $G_n$ contains the universal vertex $v$ such that the class $\{G_n-v:n\ge 1\}$ has bounded treewidth layerings;
and for many optimization problems, one can devise a reduction from the problem in $G$ to a variant of the problem in $G-v$
(possibly encoding the neighborhood of $v$ by coloring vertices of $G-v$).
\item The grids $U_n$ have the property that in their distance layering, the unions of bounded numbers of layers induce subgraphs
which themselves have bounded treewidth layerings, making it possible to iterate Baker's technique.
\end{itemize}
Let $\CC$ be a class of graphs.  Informally, we will say that $\CC$ is a \emph{Baker class} if
each graph from $\CC$ can be reduced to an empty graph by a constant number of iterations of these operations
(removal of vertices, arbitrary choice of a bounded number of consecutive layers in a layering).
To enable our intended application to proper minor-closed classes,
we need to allow the number of iterations depend not only on $\CC$, but also on the number of layers we select
from each layering (which in itself depends on the considered optimization problem and the desired precision $\varepsilon$
of the approximation).  This becomes problematic if the number of iterations is large compared to $1/\varepsilon$,
as the errors accumulate in each iteration.  To deal with this issue, we allow the number of layers selected from
the layering to grow in each iteration.  As the resulting definition is rather technical, we postpone it to Section~\ref{sec-bakergame}.
Let us remark that the iteration in the definition of a Baker class ends in an empty graph
(rather than a graph of bounded treewidth).  Stopping when we reach graphs of bounded treewidth would not result in a more
general property, since classes with bounded treewidth are themselves Baker.

Building upon the work of Dawar et al.~\cite{dawar2006approximation}, we show that monotone optimization problems expressible in the first-order
logic admit Polynomial-Time Approximation Schemes on Baker classes.
Throughout the paper, we work with first-order formulas on graphs, i.e., formulas using a single irreflexive symmetric
binary predicate $e$ (interpreted by the adjacency in a graph), any number of unary predicates (interpreted as colors on vertices of the graph),
quantification on variables for vertices of the graph, equality, and standard logic operators $\land$, $\lor$, and $\neg$ (with
other operators such as $\Rightarrow$ expressed in terms of these basic operators).
A \emph{graph language} $L$ consists of the binary predicate symbol $e$ and a finite set of unary predicate symbols.
An \emph{$L$-interpretation} $G$ consists of a graph $\overline{G}$ and a set $S_C$ of vertices of $G$ for each unary symbol $C\in L$;
the binary symbol $e$ is naturally interpreted as the set of all pairs of adjacent vertices of $\overline{G}$,
while for each unary symbol $R$ and vertex $v\in V(\overline{G})$, we have $G,x\colonequals v\models C(x)$ if and only if $v\in S_C$.
Otherwise, the semantics of first-order formulas is defined in the usual way.

Consider a first-order formula $\varphi$ using a unary predicate $X$.  We say that $\varphi$ is \emph{$X$-positive}
if all appearances of $X$ in $\varphi$ are within the scope of an even number of negations, and \emph{$X$-negative}
if all appearances of $X$ in $\varphi$ are within the scope of an odd number of negations.
Suppose $A\subseteq B$; note that if $\varphi$ is $X$-positive, then $X\colonequals A\models \varphi$ implies $X\colonequals B\models \varphi$;
and if $\varphi$ is $X$-negative $X\colonequals B\models \varphi$ implies $X\colonequals A\models \varphi$.

A \emph{positive FO minimization problem} given by an $X$-positive first-order sentence $\varphi$ over a graph language $L\cup\{X\}$
seeks, for an input $L$-interpretation $G$, to find
a set $A\subseteq V(\overline{G})$ of the minimum size such that $G,X\colonequals A\models \varphi$.  Let the minimum size of such set
be denoted by $\gamma_{\varphi}(G)$.  The basic example is the \emph{minimum dominating set} problem, given by the sentence
$$(\forall x) \bigl[X(x)\lor (\exists y)(e(x,y)\land X(y))\bigr];$$
other examples include the minimum vertex cover and minimum-size set intersecting all triangles.

A \emph{negative FO maximization problem} given by an $X$-negative first-order sentence $\varphi$ over a graph language $L\cup\{X\}$
seeks, for an input $L$-interpretation $G$, to find
a set $A\subseteq V(\overline{G})$ of the maximum size such that $G,X\colonequals A\models \varphi$.  Let the maximum size of such set
be denoted by $\alpha_\varphi(G)$.  The basic example is the \emph{maximum independent set} problem, given by the sentence
$$(\forall x,y) e(x,y)\Rightarrow (\neg X(x)\lor\neg X(y));$$
other examples include the maximum $r$-independent set for any fixed integer $r$ (i.e., the largest subset of vertices
at distance greater than $r$ from each other) and largest induced subgraph of maximum degree at most $d$ for any fixed integer $d$.

Note that if $\varphi'$ is the $X$-positive sentence obtained from an $X$-negative sentence $\varphi$ by negating
every appearance of $X$, then $\alpha_\varphi(G)=|V(\overline{G})|-\gamma_{\varphi'}(G)$; nevertheless, in case that
$\alpha_\varphi(G)=o(|V(\overline{G})|)$, even a very precise approximation of $\gamma_{\varphi'}(G)$ does not give a good
approximation of $\alpha_\varphi(G)$, and thus from the approximation perspective, the two problems are distinct.
We give a PTAS for both variants (see Section~\ref{ssec-alg} for the definition of an \emph{$(f,s)$-efficiently Baker class}).

\begin{theorem}\label{thm-approx}
Suppose $\CC$ is an $(f,s)$-efficiently Baker class of graphs.
There exists an algorithm that given a first-order $X$-positive sentence $\varphi$ over a graph language $L\cup \{X\}$,
an integer $k$ and an $L$-interpretation $G$ with $n$ vertices
such that $\overline{G}\in\CC$, in time $f(n)+O_{\CC,\varphi,k}(n+s(n))$ finds a set $A\subseteq V(\overline{G})$
satisfying
$$G,X\colonequals A\models \varphi$$
and $|A|\le (1+1/k)\gamma_\varphi(G)$, or determines no such set $A$ exists.
\end{theorem}

\begin{theorem}\label{thm-approx-max}
Suppose $\CC$ is an $(f,s)$-efficiently Baker class of graphs.
There exists an algorithm that given a first-order $X$-negative sentence $\varphi$ over a graph language $L\cup \{X\}$,
an integer $k$ and an $L$-interpretation $G$ with $n$ vertices
such that $\overline{G}\in\CC$, in time $f(n)+O_{\CC,\varphi,k}(n+s(n))$ finds a set $A\subseteq V(\overline{G})$
satisfying
$$G,X\colonequals A\models \varphi$$
and $|A|\ge (1-1/k)\alpha_\varphi(G)$, or determines no such set $A$ exists.
\end{theorem}
As we alluded to before, we prove that any proper minor-closed class of graphs is $(O(n^2),O(n))$-efficiently Baker,
and thus Theorems~\ref{thm-approx} and \ref{thm-approx-max} give PTASes with time complexity $O_{\CC,\varphi,k}(n^2)$
for proper minor-closed classes.
Unlike us, Dawar et al.~\cite{dawar2006approximation} only give algorithms with time complexity $n^{O_{\CC,\varphi,k}(1)}$.
More importantly, their argument uses the minor structure theorem~\cite{robertson2003graph}, and thus it is specific
only to proper minor-closed classes and the multiplicative constants hidden in the O-notation are huge.
Our results are in the setting of more general Baker classes, and the proof that proper minor-closed classes are Baker
is direct, without using the minor structure theorem; consequently, the multiplicative constants of the O-notation are
more manageable (although still impractically large).

The paper is organized as follows: In Section~\ref{sec-bakergame}, we give the definition of Baker classes via Baker games introduced in
Section~\ref{sec-game}, and work out the properties of this concept, culminating in the proof that proper minor-closed classes
are Baker in Section~\ref{sec-minors}.  We discuss the algorithmic version of the concept in Section~\ref{ssec-alg},
and another Baker class in Section~\ref{sec-distort}.  In Section~\ref{sec-ptas}, we give polynomial-time approximation schemes
for Baker classes: we prove Theorem~\ref{thm-approx} in Section~\ref{sec-positive}, Theorem~\ref{thm-approx-max} 
in Section~\ref{sec-negative}, and discuss problems not expressible in the first-order logic in Section~\ref{sec-other}.

\subsection{Related results}

As we already mentioned in the introduction, Eppstein~\cite{eppstein00} and Demaine and Hajiaghayi~\cite{demaine2004equivalence}
showed that Baker's technique generalizes
to all proper minor-closed classes that do not contain all apex graphs.  Going beyond the apex graph boundary, Grohe~\cite{grohe2003local},
Demaine et al.~\cite{demaine2005algorithmic} and Dawar et al.~\cite{dawar2006approximation} generalized
the approximation algorithms to all proper minor-closed classes using the tree decomposition from the minor structure theorem.
Baker's technique also applies to other geometrically defined graph classes, such as unit disk graphs~\cite{hunt1998nc}
and graphs embedded with a bounded number of crossings on each edge~\cite{grigoriev2007algorithms}.

To go beyond classes with bounded treewidth layerings, Dvo\v{r}\'ak~\cite{twd} introduced a weaker notion of
\emph{fractional treewidth-fragility}: instead of requiring the existence of a layering where consecutive layers induce subgraphs
of bounded treewidth, fractional treewidth-fragility only requires the existence of many nearly-disjoint subsets of vertices
whose removals results in a graph of bounded treewidth.  This notion is sufficient to obtain polynomial-time approximation schemes for some graph parameters
such as the independence number or the size of the largest $H$-matching, but fails for others (minimum dominating set,
distance constrained versions of the independence number). On the other hand, all proper minor-closed classes are
fractionally treewidth-fragile~\cite{devospart}, and so are all subgraph-closed graph classes with bounded maximum degree
and strongly sublinear separators~\cite{twd}.  Dvořák~\cite{thin} later introduced a stronger notion of
\emph{thin systems of overlays}, also applicable to these classes (via the minor structure theorem)
and sufficient to obtain PTASes for more problems (although not all that can be handled using our approach).
Note it is easy to see every Baker class has thin systems of overlays, and in particular this gives a proof
that proper minor-closed classes have thin systems of overlays and are fractionally treewidth-fragile without
using the minor structure theorem.

Another algorithmic approach for proper minor-closed graph classes
is through the bidimensionality theory, bounding the treewidth of the graph in terms of the size of the optimal solution
and exploiting the arising bounded-size balanced separators to obtain approximate solutions.
Demaine and Hajiaghayi~\cite{demaine2005bidimensionality} and Fomin et al.~\cite{fomin2011bidimensionality}
use this approach to construct polynomial-time approximation schemes on all proper minor-closed classes
for many minor-monotone problems (e.g., minimum vertex cover) and on apex-minor-free classes for contraction-monotone
problems (e.g., minimum dominating set).

A very different approach is taken by Cabello and Gajser~\cite{cabello2015simple} for proper minor-closed classes and more generally by
Har-Peled and Quanrud~\cite{har2015approximation} for classes of graphs with polynomial expansion (which by the
result of Dvo\v{r}\'ak and Norin~\cite{dvorak2016strongly} is equivalent to having strongly sublinear separators).
They showed that the trivial local search algorithm (performing bounded-size changes on an initial solution as long as it can
be improved by such a change) gives polynomial-time approximation schemes for 
maximum independent and minimum dominating set, as well as many other related problems.
It is an open problem whether some variation on this approach can give PTAS for all monotone FO optimization problems.

Let us remark that since the approximation factor does not affect the exponent in the complexity in Theorems~\ref{thm-approx}
and \ref{thm-approx-max}, we also obtain fixed-parameter tractability for all the considered problems when parameterized by
the order of the optimum solution.  However, the fixed-parameter tractability of these problems has been established
in greater generality, see~\cite{dvorak2013testing,grohe2014deciding}.
Our approach was in part inspired by~\cite{grohe2014deciding}; to obtain their result, they introduce Splitter games, and a winning
strategy for a Baker game translates into a winning strategy for the Splitter game (but not vice versa).

\section{Baker game}\label{sec-bakergame}

\subsection{Layerings and ordered graphs}\label{sec-prelim}

Let us start with some preliminaries.

Let $G$ be a graph and consider a function $\lambda:V(G)\to\mathbb{Z}$.
If $G$ is connected, $v_0\in V(G)$, and $\lambda(v)$ is equal to the distance
from $v_0$ to $v$ in $G$ for each $v\in V(G)$, then $\lambda$ is a layering; we call this layering the \emph{BFS layering starting from $v_0$}.
If $G$ has components $G_1$, \ldots, $G_m$ and for some positive integer $r$, $\lambda(v)=ri$ holds for every $i\in\{1,\ldots,m\}$ and $v\in V(G_i)$,
then $\lambda$ also is a layering, which we call the \emph{$r$-spread componentwise layering}.
For a layering $\lambda$, the \emph{width} of $\lambda$ is defined as $\sup\{|\lambda^{-1}(i)|:i\in\mathbb{Z}\}$.

For any two vertices $x$ and $y$ of a graph $G$, let $d_G(x,y)$ denote the distance between $x$ and $y$ in $G$.
For $P\subseteq V(G)$, a layering $\lambda$ of $G[P]$ is \emph{$G$-geodesic} if $d_G(x,y)\ge |\lambda(x)-\lambda(y)|$ for every $x,y\in P$.
We need the following observation on extendability of layerings.
\begin{lemma}\label{lemma-exlay}
Let $G$ be a graph and let $P\subseteq V(G)$. A layering $\lambda$ of $G[P]$
extends to a layering of $G$ if and only if $\lambda$ is $G$-geodesic.
\end{lemma}
\begin{proof}
Any layering $\lambda'$ of $G$ satisfies $d_G(x,y)\ge |\lambda'(x)-\lambda'(y)|$ for every $x,y\in V(G)$; hence,
if $\lambda$ extends to a layering of $G$, then $\lambda$ is $G$-geodesic.

Let us now argue that every $G$-geodesic layering of $G[P]$ extends to a layering of $G$.
We prove the claim by induction of $|V(G)\setminus P|$. The case $V(G)=P$ being trivial, we can assume
there exists a vertex $v\in V(G)\setminus P$.
Let us define $\lambda(v)=\max\{\lambda(x)-d_G(x,v):x\in P\}$.  By definition, we have $\lambda(v)\ge \lambda(x)-d_G(x,v)$,
and thus $d_G(x,v)\ge \lambda(x)-\lambda(v)$ for every $x\in P$.  Also, there exists $y\in P$ such that $\lambda(v)=\lambda(y)-d_G(y,v)$.
Consider any $x\in P$; by the triangle inequality we have $d_G(y,v)\ge d_G(x,y)-d_G(x,v)$, and since $\lambda$ is $G$-geodesic, we conclude
$\lambda(v)-\lambda(x)=\lambda(y)-\lambda(x)-d_G(y,v)\le \lambda(y)-\lambda(x)-d_G(x,y)+d_G(x,v)\le d_G(x,v)$.
Consequently, with this definition of $\lambda(v)$, we have $d_G(x,y)\ge |\lambda(x)-\lambda(y)|$ for every $x,y\in P\cup \{v\}$.
In particular, if $xv\in E(G)$ for $x\in P$, then $|\lambda(x)-\lambda(v)|\le 1$, and thus $\lambda$ is a $G$-geodesic layering of $G[P\cup\{v\}]$.
By the induction hypothesis, $\lambda$ extends to a layering of $G$.
\end{proof}

An \emph{ordered graph} is a graph together with a linear ordering of its vertices.  If $G'$ is a subgraph of an ordered graph $G$,
the ordering of vertices of $G'$ is naturally the ordering of vertices of $G$ restricted to $V(G')$.

\subsection{Rules of the game}\label{sec-game}

For an infinite sequence $\seq{r}=r_1,r_2,\ldots$ and an integer $s\ge 0$, let $\tail_s(\seq{r})$ denote the sequence $r_{s+1}, r_{s+2}, \ldots$,
let $\tail(\seq{r})=\tail_1(\seq{r})$, and let $\head(\seq{r})=r_1$.
A sequence $\seq{r'}=r'_1,r'_2,\ldots$ is \emph{dominated by $\seq{r}$} if $r'_i\le r_i$ for every $i\in\mathbb{N}$.

The \emph{Baker game} between two players (Destroyer and Preserver) is defined as follows.
The state of the game is a pair $(G,\seq{r})$, where $G$ is an ordered graph and $\seq{r}$ is an infinite
sequence. If $V(G)=\emptyset$, then the game stops.  Otherwise, Destroyer chooses one of the following actions:
\begin{itemize}
\item[\textbf{Delete}] Destroyer deletes the smallest vertex $v$ from the graph; Preserver takes no action and
the game proceeds with the state $(G-v,\tail(\seq{r}))$.
\item[\textbf{Restrict}] Destroyer selects a layering $\lambda$ of $G$.  Preserver selects an interval $I$ of at most $\head(\seq{r})$
consecutive integers and the game proceeds with the state $(G[\lambda^{-1}(I)],\tail(\seq{r}))$.
That is, Preserver selects $\head(\seq{r})$ consecutive layers
and deletes the rest of the graph.
\end{itemize}
Destroyer seeks to minimize the number of rounds of this game; we say that Destroyer \emph{wins in $t$ rounds} on the state $(G,\seq{r})$ if
regardless of Preserver's strategy, the game stops after at most $t$ rounds.

Let us remark that in the Delete action, we could more generally
allow the Destroyer to delete any vertex, not just the smallest one; however, throughout this paper we only deal with the ``monotone'' strategies
where the vertices are deleted in order, and it is more convenient to formulate the game already including this assumption rather than repeating
it everywhere. Furthermore, since the Delete action does not depend on the sequence $\seq{r}$, it is tempting not to consume its
first element, and proceed with the state $(G-v,\seq{r})$ rather than $(G-v,\tail(\seq{r}))$.  However, in that setting we would run into
problems in some of the arguments below, in particular in the proof of Lemma~\ref{lemma-compose}.

\subsection{Basic properties}

We will often use the following observation.
\begin{lemma}\label{lemma-sg}
Suppose Destroyer wins the Baker game on the state $(G,\seq{r})$ in $t$ rounds.
If $G'$ is a subgraph of $G$ and $\seq{r}$ dominates a sequence $\seq{r'}$, then
Destroyer wins the Baker game on the state $(G',\seq{r'})$ in $t$ rounds.
\end{lemma}
\begin{proof}
We prove the claim by induction on $t$.  If $V(G')=\emptyset$, then Destroyer wins on the state $(G',\seq{r'})$ in $0\le t$ rounds.
Hence, suppose that $V(G')\neq\emptyset$, and thus also $V(G)\neq \emptyset$.
Since Destroyer wins on the state $(G,\seq{r})$ in $t$ rounds, we have $t\ge 1$.
We consider the first action of Destroyer on $(G,\seq{r})$.

If he takes the Delete action, then on $(G',\seq{r'})$ we also take the delete action.  Let $v$ and $w$ be the smallest
vertices of $G$ and $G'$, respectively.  Note that if $v\in V(G')$, then $v=w$; consequently, we have $G'-w\subseteq G-v$.
By the induction hypothesis Destroyer wins on $(G'-w,\tail(\seq{r'}))$ in $t-1$ rounds, and thus he wins on $(G',\seq{r'})$ in $t$ rounds.

If Destroyer selects a layering $\lambda$ of $G$, we select a layering $\lambda'=\lambda\restriction V(G')$ of $G'$.
The Preserver answers by choosing an interval $I$ of at most $\head(\seq{r'})\le \head(\seq{r})$ consecutive integers.
Destroyer wins on the state $(G[\lambda^{-1}(I)],\tail(\seq{r}))$ in $t-1$ rounds, and by the induction hypothesis,
he also wins on the state $(G'[\lambda^{-1}(I)],\tail(\seq{r'}))$ in $t-1$ rounds.  Consequently, Destroyer wins on
$(G',\seq{r'})$ in $t$ rounds.
\end{proof}

We say that $\CC$ is a \emph{Baker class} if $\CC$ is a class of ordered graphs and for every sequence $\seq{r}$ there exists an integer
$t$ such that for each $G\in\CC$, Destroyer wins the Baker game on $(G,\seq{r})$ in $t$ rounds.

Let us now prove an important composition result.
A \emph{partition} $\PP$ of an ordered graph $G$ (with linear ordering $\prec$ of its vertices) is a sequence $P_1$, \ldots, $P_m$
of pairwise disjoint subsets of $V(G)$ such that $V(G)=P_1\cup\ldots\cup P_m$ and all vertices $u\in P_i$ and $v\in P_j$ such that $i<j$
satisfy $u\prec v$.  
For an integer $d\ge 1$, the partition $\PP$ is \emph{width-$d$ geodesic} if for $i\in\{1,\ldots,m\}$, $G[P_i]$ has
a $G[P_i\cup P_{i+1}\cup \ldots\cup P_m]$-geodesic layering of width at most $d$.
Let $G/\PP$ denote the ordered graph obtained from $G$ by identifying all vertices in each part $P$ of $\PP$ to a single vertex and
suppressing the arising loops and parallel edges, with the ordering of the vertices of $G/\PP$ matching
the sequence of parts of $\PP$.
For a class $\CC$ and an integer $d\ge 1$, let $\CC^{(d)}$ denote the class of ordered graphs $G$
for which there exists a width-$d$ geodesic partition $\PP$ of $G$ such that $G/\PP\in\CC$.

\begin{lemma}\label{lemma-compose}
If $\CC$ is Baker class, then $\CC^{(d)}$ is a Baker class for every integer $d\ge 1$.
\end{lemma}
\begin{proof}
By Lemma~\ref{lemma-sg}, we can without loss of generality assume that $\CC$ is subgraph-closed.
For any sequence $\seq{r}$, let $t(\seq{r})$ be an integer such that for every $H\in\CC$, Destroyer wins
the Baker game on $(H,\seq{r})$ in $t(\seq{r})$ rounds.

Consider any sequence $\seq{r}=r_1,r_2,\ldots$.  Let $i_0=0$, and for $j\ge 1$, let
$i_j=i_{j-1}+dr_{i_{j-1}+1}+1$.
Let $\seq{r'}=r_{i_0+1},r_{i_1+1},\ldots$ and let $m=t(\seq{r'})$.
We claim that for every ordered graph $G\in \CC^{(d)}$, Destroyer wins the Baker game on $(G,\seq{r})$ in $i_m$ rounds.
More generally, we will prove the following claim for $n=0,\ldots, m$ by induction.  Suppose that $H\in\CC$ is an ordered graph
such that Destroyer wins the Baker game on $(H,\tail_{m-n}(\seq{r'}))$ in $n$ rounds.  If $G\in\CC^{(d)}$
has a width-$d$ geodesic partition $\PP$ such that $H=G/\PP$, then Destroyer wins the Baker game
on $(G,\tail_{i_{m-n}}(\seq{r}))$ in $i_m-i_{m-n}$ rounds.

We identify the vertex set of $H$ with $\PP$ in the natural way.
Note that $$\head(\tail_{m-n}(\seq{r'}))=r_{i_{m-n}+1}=\head(\tail_{i_{m-n}}(\seq{r})).$$
The claim is trivial for $n=0$, and thus we can assume $n\ge 1$.  We mimic the Destroyer's strategy for $(H,\tail_{m-n}(\seq{r'}))$ as follows.

If Destroyer performs the Restrict action with layering $\lambda$,
we let $\lambda_\star$ denote the layering of $G$
such that $\lambda_\star(v)=\lambda(P)$ for every $P\in V(H)$ and $v\in P$.
In the Baker game on $(G,\tail_{i_{m-n}}(\seq{r}))$, we perform the Restrict action with the layering $\lambda_\star$,
Preserver chooses an interval $I$
of at most $r_{i_{m-n}+1}$ consecutive integers and changes the state to
$(G_\star,\tail_{i_{m-n}+1}(\seq{r}))$, where $G_\star=G[\lambda^{-1}_\star(I)]$.
We follow up with $i_{m-n+1}-i_{m-n}-1$ Delete actions, either ending the game in the process
or changing the state to $(G'_\star,\tail_{i_{m-n+1}}(\seq{r}))$ for a subgraph $G'_\star$ of $G_\star$.
In the Baker game on $(H,\tail_{m-n}(\seq{r'}))$, we have Preserver also answer with $I$, resulting in the state
$(H',\tail_{m-n+1}(\seq{r'}))$, where $H'=H[\lambda^{-1}(I)]$.
Considering the partition $\PP_\star=V(H')\subseteq\PP$, observe that $H'=G_\star/\PP_\star$ and that $\PP_\star$ is a width-$d$ geodesic partition
of $G_\star$.  Since Destroyer wins the Baker game on $(H',\tail_{m-n+1}(\seq{r}))$ in $n-1$ rounds,
by the induction hypothesis and Lemma~\ref{lemma-sg} Destroyer wins the Baker game on $(G'_\star,\tail_{i_{m-n+1}}(\seq{r}))$ in
$i_m-i_{m-n+1}$ rounds, and thus Destroyer wins the Baker game on
$(G,\tail_{i_{m-n}}(\seq{r}))$ in $1+(i_{m-n+1}-i_{m-n}-1)+(i_m-i_{m-n+1})=i_m-i_{m-n}$ rounds.

Next, suppose Destroyer performs the Delete action in the Baker game on $(H,\tail_{m-n}(\seq{r'}))$,
changing the state to $(H-P,\tail_{m-n+1}(\seq{r'})$, where $P\in\PP$ is the smallest vertex of $H$.
We select a $G$-geodesic layering $\lambda_P$ of $G[P]$ of width at most $d$.  By Lemma~\ref{lemma-exlay},
$\lambda_P$ extends to a layering $\lambda$ of $G$.  We perform the Restrict action on $(G,\tail_{i_{m-n}}(\seq{r}))$
with this layering $\lambda$.  Preserver answers with an interval $I$ of at most $r_{i_{m-n}+1}$ consecutive integers, changing the state
to $(G[\lambda^{-1}(I)],\tail_{i_{m-n}+1}(\seq{r}))$.  Note that $|\lambda^{-1}(I)\cap P|\le dr_{i_{m-n}+1}$.
Next, we perform $dr_{i_{m-n}+1}$ Delete actions, either ending the game in the process or
changing the state to $(G',\tail_{i_{m-n+1}}(\seq{r}))$, where
$G'\subseteq G[\lambda^{-1}(I)\setminus P]\subseteq G[V(G)\setminus P]$.
Note that $H-P=G[V(G)\setminus P]/(\PP\setminus \{P\})$,
and $\PP\setminus \{P\}$ is a width-$d$ geodesic partition of $G[V(G)\setminus P]$.  Since Destroyer wins the Baker game on
$(H-P,\tail_{m-n+1}(\seq{r'}))$ in $n-1$ rounds, the induction hypothesis and Lemma~\ref{lemma-sg} implies that
Destroyer wins the Baker game on $(G',\tail_{i_{m-n+1}}(\seq{r}))$ in $i_m-i_{m-n+1}$ round, and
thus he also wins the Baker game on $(G,\tail_{i_{m-n}}(\seq{r}))$ in
$1+dr_{i_{m-n}+1}+(i_m-i_{m-n+1})=i_m-i_{m-n}$ rounds.
\end{proof}

We will also need another composition result based on clique-sums.  Suppose $\CC_1$ and $\CC_2$ are classes of ordered graphs.
Let $\CC_1\oplus\CC_2$ denote the class of ordered graphs $G$ such that there exists a set $B\subseteq V(G)$ satisfying the following
conditions:
\begin{itemize}
\item $G[B]\in \CC_1$, and
\item for each component $C$ of $G-B$ we have $G[C]\in\CC_2$, the neighbors of vertices of $C$ in $B$
induce a clique $K_C$, and all vertices of $K_C$ are smaller than all vertices of $C$.
\end{itemize}
In this situation, we say that $B$ is the \emph{base} of $G$.

\begin{lemma}\label{lemma-csum}
If $\CC_1$ and $\CC_2$ are Baker classes, then $\CC_1\oplus \CC_2$ is a Baker class.
\end{lemma}
\begin{proof}
Consider a sequence $\seq{r}=r_1,r_2,\ldots$, and let $\seq{r'}=r_2,r_4,r_6,\ldots$.  Let $t_1$ be an integer such that
for every $H\in\CC_1$, Destroyer wins the Baker game on $(H,\seq{r'})$ in $t_1$ rounds.
Let $t_2$ be an integer such that for every $H\in\CC_2$, Destroyer wins the Baker game on
$(H,\tail_{2t_1+1}(\seq{r}))$ in $t_2$ rounds.
By Lemma~\ref{lemma-sg}, we can without loss of generality assume that $\CC_1$ and $\CC_2$ are subgraph-closed.

We claim that for every ordered graph $G\in \CC_1\oplus \CC_2$, Destroyer wins the Baker game on $(G,\seq{r})$ in $2t_1+t_2+1$ rounds.
More generally, we will prove the following claim for $n=0,\ldots, t_1$ by induction.  Consider a graph $G\in\CC_1\oplus \CC_2$ with base $B$.
If Destroyer wins the Baker game on $(G[B],\tail_{t_1-n}(\seq{r'}))$ in $n$ rounds, then he also wins the Baker game
on $(G,\tail_{2(t_1-n)}(\seq{r}))$ in $2n+t_2+1$ rounds.

Destroyer first performs the Restrict action with the $r_{2(t_1-n)+1}$-spread componentwise layering.
Preserver's response then changes the state to $(G',\tail_{2(t_1-n)+1}(\seq{r}))$, where either $V(G')=\emptyset$ or $G'$ is a component of $G$;
we can assume the latter.  Since $\CC_1$ and $\CC_2$ are subgraph-closed, $B'=B\cap V(G')$ is a base of $G'$.

If $B'=\emptyset$, then since $B'$ is a base of $G'$ and $G'$ is connected, we conclude that $G'\in\CC_2$.
Destroyer performs (at most) $2n$ Delete actions, resulting either in a victory or a state $(G'',\tail_{2t_1+1}(\seq{r}))$
with $G''\in \CC_2$.  In the latter case, Destroyer wins the Baker game on $(G'',\tail_{2t_1+1}(\seq{r}))$ in $t_2$ rounds.
Hence, Destroyer wins the Baker game on $(G,\tail_{2(t_1-n)}(\seq{r}))$ in $2n+t_2+1$ rounds, as required.
Hence, we can assume $B'\neq\emptyset$.  In particular, we have $n\ge 1$, since when $n=0$,
Destroyer wins the Baker game on $(G[B],\tail_{t_1-n}(\seq{r'}))$ in $0$ rounds, and thus $B=\emptyset$,

Since $G'[B']=G[B']\subseteq G[B]$, Lemma~\ref{lemma-sg} implies that Destroyer wins the Baker game on
$(G'[B'],\tail_{t_1-n}(\seq{r'}))$ in $n$ rounds.  Note that $\head(\tail_{t_1-n}(\seq{r'}))=\head(\tail_{2(t_1-n)+1}(\seq{r}))=r_{2(t_1-n)+2}$,
and consider the first action of Destroyer on this state.

If the action is Restrict with layering $\lambda_\star$, then let $\lambda:V(G')\to\mathbb{Z}$ be defined as follows.
Let $\lambda(v)=\lambda_\star(v)$ for $v\in B'$.  For every component $C$ of $G'-B'$, choose $z\in B'$ with a neighbor in $C$
(which exists since $G'$ is connected) arbitrarily and let $\lambda(v)=\lambda_\star(z)$ for every $v\in C$.
Note that $\lambda$ is a layering of $G'$: if $z'\in B'\setminus\{z\}$ has a neighbor $v\in C$, then since the vertices with a neighbor in $C$
form a clique, we have $zz'\in E(G')$ and $|\lambda(z')-\lambda(v)|=|\lambda_\star(z')-\lambda_\star(z)|\le 1$.
On the state $(G',\tail_{2(t_1-n)+1}(\seq{r}))$, we preform the Restrict action with this layering $\lambda$ and Preserver
answers with an interval $I$ of at most $r_{2(t_1-n)+2}$ consecutive integers, resulting in the state
$(G'[\lambda^{-1}(I)],\tail_{2(t_1-n+1)}(\seq{r}))$.  Destroyer wins the Baker game
on $(G'[\lambda_\star^{-1}(I)\cap B'],\tail_{t_1-n+1}(\seq{r'}))$ in $n-1$ rounds, and $\lambda_\star^{-1}(I)\cap B'$ is a base
of $G'[\lambda^{-1}(I)]$.  By the induction hypothesis, Destroyer wins the Baker game on $(G'[\lambda^{-1}(I)],\tail_{2(t_1-n+1)}(\seq{r}))$
in $2(n-1)+t_2+1$ rounds, and thus he also wins on $(G,\tail_{2(t_1-n)}(\seq{r}))$ in $2n+t_2+1$ rounds.

Suppose now the action is Delete, hence changing the state to $(G'[B'\setminus \{v\}],\tail_{t_1-n+1}(\seq{r'}))$, where $v$ is the
smallest vertex of $B'$; Destroyer
wins the Baker game from this state in $n-1$ rounds. Note that $v$ is also the smallest vertex of $G'$, since $B'$ is a base of $G'$ and
$G'$ is connected.  Hence, the Delete action applied to $(G',\tail_{2(t_1-n)+1}(\seq{r}))$ deletes the same vertex, resulting in the state
$(G'-v,\tail_{2(t_1-n+1)}(\seq{r}))$.  Observe that $B'\setminus \{v\}$ is a base of $G'-v$. By the induction hypothesis,
Destroyer wins the baker game on $(G'-v,\tail_{2(t_1-n+1)}(\seq{r}))$ in $2(n-1)+t_2+1$ rounds, and thus he also wins on $(G,\tail_{2(t_1-n)}(\seq{r}))$ in $2n+t_2+1$ rounds.
\end{proof}

\subsection{Bounded treewidth}

An ordered graph $G$ is \emph{chordal} if for every vertex $v\in V(G)$, the neighbors of $v$ smaller than $v$ induce a clique in $G$.
The \emph{left-degree} of $v$ is the number of such neighbors.  Note that vertices of a graph $H$ can be linearly ordered so that
the resulting ordered graph is chordal if and only if $H$ is chordal, i.e., contains no induced cycles of length at least $4$.
Furthermore, recall that graph $H'$ has treewidth at most $d$ if and only if $H'$ has a chordal supergraph $H$ with clique number
at most $d+1$, and thus the corresponding chordal ordered graph has maximum left-degree at most $d$.

Let us start with a simple observation.
\begin{lemma}\label{lemma-mop}
Let $G$ be a chordal ordered graph (with linear ordering $\prec$ of its vertices).
If $P=v_1v_2\ldots v_k$ is an induced path in $G$ from $v_1$ to $v_k$ and $v_1$
is smaller than all other vertices of $P$, then $v_1\prec \ldots\prec v_k$.
\end{lemma}
\begin{proof}
Suppose for a contradiction there exists $s\in\{2,\ldots, k-1\}$ such that $v_{s+1}\prec v_s$, and choose smallest such index $s$.
By the minimality of $s$, we have $v_{s-1}\prec v_s$.  Since $G$ is chordal, $v_{s-1}$ and $v_{s-1}$ are adjacent,
contradicting the assumption that $P$ is an induced path.
\end{proof}

For a non-negative integer $r$ and a vertex $v$ of a graph $G$, let $N_r[v]$ denote the set of vertices of $G$ at distance at most $r$ from $v$.

\begin{lemma}\label{lemma-comp}
Let $G$ be a chordal ordered graph, let $v$ be the smallest vertex of $G$,
let $r$ be a non-negative integer, let $C$ be the vertex set of a component of $G-N_r[v]$,
and let $K$ be the set of vertices in $N_r[v]$ which have a neighbor in $C$.
Then $K$ induces a clique in $G$, and all vertices in $K$ are smaller than all vertices in $C$.
\end{lemma}
\begin{proof}
Suppose first that distinct vertices $x,y\in K$ are not adjacent.  Note that $d_G(v,x)=d_G(v,y)=r$, as otherwise $C$ would contain
a vertex at distance at most $r$ from $v$.  Let $P_1$ be the shortest path between $x$ and $y$ in $G[N_{r-1}[v]\cup\{x,y\}]$,
and let $P_2$ be a shortest path between $x$ and $y$ in $G[C\cup\{x,y\}]$.  The concatenation of $P_1$ and $P_2$ is
an induced cycle of length at least $4$, contradicting the assumption that $G$ is chordal.  Hence, $G[K]$ is a clique.

Let $x$ be the largest vertex in $K$, and let $z$ be the smallest vertex in $C$.
Suppose now for a contradiction that $z$ is smaller than $x$.  Let $Q$ be a shortest path between $x$ and $z$ in $G[C\cup \{x\}]$.
Note that $z$ is the smallest vertex of $V(Q)$, and by Lemma~\ref{lemma-mop}, $x$ is the largest vertex of $V(Q)$.
Let $y$ be the neighbor of $x$ in $Q$. Since $y\in C$, we conclude that $d_G(v,y)=r+1$ and there exists a shortest path
$P$ from $v$ to $y$ (of length $r+1$) containing $x$.  By Lemma~\ref{lemma-mop}, $y$ is the largest vertex of $V(P)$.
However, this is a contradiction, since $x$ and $y$ belong to both $P$ and $Q$.
\end{proof}

We are now ready to prove classes of graphs of bounded treewidth
(or more precisely, classes of chordal ordered graphs with bounded maximum left-degree)
are Baker.  For an integer $d$, let $\TT_d$ denote the class of chordal ordered
graphs of maximum left-degree at most $d$.  Let $\TT_{d,1}=\TT_d$, and for any integer $r\ge 2$, let $\TT_{d,r}=\TT_{d,r-1}\oplus\TT_d$.

\begin{lemma}\label{lemma-winchord}
For every integer $d$, $\TT_d$ is a Baker class.
\end{lemma}
\begin{proof}
We prove the claim by induction on $d$.  The class $\TT_0$ consists of ordered graphs with no edges,
and thus it is a Baker class.  Hence, we can assume $d\ge 1$.

Consider any sequence $\seq{r}=r_1,r_2,\ldots$. By the induction hypothesis and repeated applications of Lemma~\ref{lemma-csum},
$\TT_{d-1,r_2}$ is a Baker class.  Let $t$ be an integer such that for every $H\in\TT_{d-1,r_2}$,
Destroyer wins the Baker game on $(H,\tail_2(\seq{r}))$ in $t$ rounds.  We claim that for every $G\in\TT_d$,
Destroyer wins the Baker game on $(G,\seq{r})$ in $t+2$ rounds, by the following strategy.

First, Destroyer performs the Restrict action with the $r_1$-spread componentwise layering; Preserver's answer results
in a state $(G',\tail(\seq{r}))$, where $G'\in\TT_d$ is connected.  Next, Destroyer performs the Restrict action with
the BFS layering $\lambda$ starting from the smallest vertex $v$ of $G'$, and Preserver answers with an interval $I$ of
at most $r_2$ consecutive integers, changing the state to $(G'',\tail_2(\seq{r}))$, where $G''=G'[\lambda^{-1}(I)]$.

Let $b$ be the smallest element of $I$, and for $k\ge 1$, let $I_k=\{b,b+1,\ldots,b+k-1\}$.
By induction on $k\le r_2$, we prove that $G'[\lambda^{-1}(I_k)]\in\TT_{d-1,k}$.  
Let $j=b+k-1$ and consider the graph $G'_j=G'[\lambda^{-1}(j)]$.  If $j=0$, then $G'_j$ has only one vertex
and $G'_j\in \TT_{d-1}$.  If $j\ge 1$, then note that each vertex $v\in V(G'_j)$ has in $G'$ a neighbor $w$ in $\lambda^{-1}(j-1)$,
and by Lemma~\ref{lemma-comp}, $w$ is smaller than $v$.  Since $G'$ has maximum left-degree at most $d$, $G'_j$ has maximum left-degree
at most $d-1$, and thus $G'_j\in \TT_{d-1}$.  In particular, the claim holds when $k=1$, and thus we can assume $k\ge 2$.
By the induction hypothesis, we have $G'[\lambda^{-1}(I_{k-1})]\in\TT_{d-1,k-1}$, and by Lemma~\ref{lemma-comp}, it follows that
$G'[\lambda^{-1}(I_k)]\in\TT_{d-1,k}$.

We conclude that $G''\in\TT_{d-1,r_2}$, and thus Destroyer wins the Baker game on $(G'',\tail_2(\seq{r}))$ in $t$ rounds.
Consequently, Destroyer wins the Baker game on $(G,\seq{r})$ in $t+2$ rounds.
\end{proof}

\subsection{Forbidden minors}\label{sec-minors}

It is now easy to show that proper minor-closed classes are Baker.
The inspection of the proof of Lemma 4.1 in~\cite{van2017generalised} gives the following.
\begin{lemma}\label{lemma-str}
For every integer $k\ge 3$, every $K_k$-minor-free graph $G$ has a linear ordering of vertices such that the corresponding
ordered graph belongs to $\TT_{k-2}^{(k-2)}$.
\end{lemma}
Let us remark that the proof of Lemma~\ref{lemma-str} is elementary and constructive, requiring at most $|V(G)|$ breadth-first
searches on $G$ to construct the required ordering and partition.
Lemmas~\ref{lemma-compose}, \ref{lemma-winchord} and \ref{lemma-str} now give our first main result.

\begin{corollary}
For every integer $k\ge 1$, there exists a Baker class $\KK_k$ such that
every $K_k$-minor-free graph together with some linear ordering of its vertices belongs to $\KK_k$.
\end{corollary}

\subsection{Algorithmic considerations}\label{ssec-alg}

In the algorithmic application we develop in the next section, it is of course important not only
that the Baker game can be won in a constant number of rounds for a particular class of graphs,
but also that the next step in the strategy can be determined efficiently.  Inspection of the proofs
shows this is the case in all the discussed situations.  Lemma~\ref{lemma-str} has algorithmic proof,
returning for an $n$-vertex $K_k$-minor-free graph the corresponding ordering, width-$(k-2)$ geodesic partition,
and layering of each of the parts in time $O(n^2)$.   Lemma~\ref{lemma-compose} then provides an explicit
description of the strategy, in terms of the strategy for $\TT_{k-2}$ which in turn is explicitly provided
by Lemmas~\ref{lemma-csum} and \ref{lemma-winchord}; the most time-demanding step in determining the current action
is the breadth-first search in Lemma~\ref{lemma-csum}, and thus the action can be determined in time $O(n)$.
Applications of Lemma~\ref{lemma-sg} throughout the proofs of other lemmas are dealt with by keeping track
of the appropriate supergraph of the currently considered graph, determining the appropriate actions in this
supergraph, and translating them to actions in the subgraph as described in the proof of Lemma~\ref{lemma-sg}.

For functions $f,s:\mathbb{N}\to\mathbb{N}$, we say that a class $\CC$ of graphs is \emph{$(f,s)$-efficiently Baker}
if there exist algorithms $\Aa_1$ and $\Aa_2$ and for every sequence $\seq{r}$ there exists an integer $t$ so that
\begin{itemize}
\item for each $G\in\CC$, the algorithm $\Aa_1$ determines in time $f(|V(G)|)$ an ordering of vertices of $G$
(and possibly other auxiliary information) so that for the resulting ordered graph,
\item the algorithm $\Aa_2$ wins the Baker game on $(G,\seq{r})$ in $t$ rounds, using time $s(|V(G)|)$ to determine
the action at each state of the game.
\end{itemize}
With this definition, the analysis presented at the beginning of this subsection implies the following.

\begin{theorem}\label{thm-minfr}
For every integer $k\ge 1$, the class of $K_k$-minor-free graphs is $(O(n^2),O(n))$-efficiently Baker.
\end{theorem}

\subsection{Graphs embedded with bounded distortion}\label{sec-distort}

Let us now give another example of Baker classes.  For a real number $\beta\ge 1$,
we say that a graph $G$ \emph{embeds in $\mathbb{R}^d$ with distortion $\beta$}
if there exists a mapping $\mu:V(G)\to \mathbb{R}^d$ such that $|\mu(u)-\mu(v)|_\infty\ge 1$
for every distinct $u,v\in V(G)$, and $|\mu(u)-\mu(v)|_\infty\le \beta$ for every $uv\in E(G)$.
For example, the $3$-dimensional grids with diagonals naturally embed in $\mathbb{R}^3$ with distortion $1$.

For $i=1,\ldots,d$, let $\pi_i:\mathbb{R}^d\to \mathbb{R}$ denote the projection to the $i$-th coordinate.
Given a graph $G$ with embedding $\mu$ in $\mathbb{R}^d$ with distortion $\beta$, we for $i=1,\ldots,d$ define a layering
$\lambda_i$ as $\lambda_i(v)=\lfloor \pi_i(\mu(v))/\beta\rfloor$.  For a sequence $\seq{r}=r_1,r_2,\ldots$, we can now apply the following
strategy: First, we perform Restrict actions with layerings $\lambda_1$, \ldots, $\lambda_d$.  Any of the remaining vertices $u$ and $v$
then satisfy $|\pi_i(u)-\pi_i(v)|\le \beta r_i$ for $i=1,\ldots, d$.  Since any distinct vertices satisfy $|\mu(u)-\mu(v)|_\infty\ge 1$,
it follows that there are at most $\prod_{i=1}^d (\beta r_i+1)$ vertices left, and we can remove them by repeated Delete actions.
Hence, Destroyer wins the Baker game in $d+\prod_{i=1}^d (\beta r_i+1)$ rounds (for any ordering of the vertices of $G$).  We conclude
the following holds.

\begin{theorem}\label{thm-distort}
Let $d\ge 1$ be an integer, let $\beta\ge 1$ be a real number, and let $\CC$ be a class of graphs such that
for every $G\in\CC$, an embedding in $\mathbb{R}^d$ with distortion $\beta$ can be found in time $f(|V(G)|)$.
Then $\CC$ is $(f,O(n))$-efficiently Baker.
\end{theorem}

\section{Polynomial-time approximation schemes}\label{sec-ptas}

\subsection{PTAS for positive FO minimization problems}\label{sec-positive}

For an integer $r\ge 0$, an \emph{$r$-local formula} is a formula $\psi$ with one free variable $x$
such that all quantifications in $\psi$ are of form $(\exists y:d(x,y)\le r)$ or $(\forall y:d(x,y)\le r)$,
where $d(x,y)\le r$ should be interpreted as the first-order formula describing that the distance between $x$ and $y$ is at most
$r$, that is, $(\exists z_0,\ldots, z_r)\, z_0=x\land z_r=y\land \bigwedge_{i=1}^r (z_{i-1}=z_i\lor e(z_{i-1},z_i))$.
That is, the validity of $\psi(x)$ depends only on the neighborhood up to distance $r$ from the vertex interpreting $x$
in the graph.

A \emph{basic existential sentence} is a sentence of form
$$(\exists x_1,\ldots,x_m) \bigwedge_{1\le i<j\le m} d(x_i,x_j)>2r \land \bigwedge_{i=1}^m \psi(x_i)$$
for some integers $m$ and $r$, where $\psi$ is $r$-local.
A \emph{basic universal sentence} is a sentence of form
$$(\forall x_1,\ldots,x_{m+1}) \bigwedge_{1\le i<j\le m+1} d(x_i,x_j)>2r \Rightarrow \bigvee_{i=1}^{m+1} \psi(x_i)$$
for some integers $m$ and $r$, where $\psi$ is $r$-local.  In both cases, $m$ is the \emph{spread},
$r$ is the \emph{range}, and $\psi$ is the \emph{core} of the sentence.
A \emph{basic conjunction} is a conjunction of basic existential and universal sentences,
and its spread and range is the maximum of spreads and ranges of these sentences.

Dawar et al.~\cite{dawar2006approximation} proved the following important result, building upon Gaifman's locality theorem~\cite{gaifman1982local}.
\begin{theorem}\label{thm-local}
Let $L$ be a graph language including a unary predicate $X$.
For every $X$-positive ($X$-negative) first-order sentence $\varphi$ over $L$, there exists an $X$-positive ($X$-negative, respectively)
first-order sentence $\theta$ over $L$ such that $\theta$ is a disjunction of basic conjunctions and for every $L$-interpretation $G$,
$G\models \varphi$ if and only if $G\models \theta$.
\end{theorem}
In view of this theorem, we can assume that every first-order sentence is given as a disjunction of basic conjunctions,
and we define its spread and range as the maximum of spreads and ranges of the terms of the disjunction.

Furthermore, we need the following observation (see also~\cite{dawar2006approximation}).
\begin{lemma}\label{lemma-remv}
Let $L\cup\{X,Z\}$ be a graph language.
For every $X$-positive ($X$-negative) first-order sentence $\varphi$ over $L\cup\{X\}$, there exist $X$-positive ($X$-negative, respectively)
first-order sentences $\varphi_1$
and $\varphi_2$ over $L\cup \{X,Z\}$ such that the following holds.  Let $G$ be an $L$-interpretation, let $v$ be a vertex of $\overline{G}$,
and let $A\subseteq V(\overline{G})$.  Then $G,X\colonequals A\models \varphi$ if and only if either
\begin{itemize}
\item $v\in A$ and $G-v,X\colonequals A\setminus\{v\},Z\colonequals N_G(v)\models \varphi_1$, or
\item $v\not\in A$ and $G-v,X\colonequals A,Z\colonequals N_G(v)\models \varphi_2$.
\end{itemize}
\end{lemma}

Let $\varphi$ be a basic existential or universal sentence of spread at most $m$, range $r$, and core $\psi$,
over graph language $L$.  Let $C,M\not\in L$ be distinct unary predicate symbols,
and let $k\le m$ be a non-negative integer.  The \emph{$(C,M,k)$-variant} of $\varphi$ is the sentence
$$\varphi^{(k)}\equiv(\exists x_1,\ldots,x_k) \bigwedge_{1\le i<j\le k} d(x_i,x_j)>2r \land \bigwedge_{i=1}^k (C(x_i)\land \psi(x_i))$$
if $\varphi$ is basic existential and
$$\varphi^{(k)}\equiv(\forall x_1,\ldots,x_{k+1}) \Bigl(\bigwedge_{i=1}^{k+1}M(x_i)\land \bigwedge_{1\le i<j\le k+1} d(x_i,x_j)>2r \Bigr)\Rightarrow \bigvee_{i=1}^{k+1} \psi(x_i)$$
if $\varphi$ is basic universal.  Note that the $(C,M,k)$-variant is $X$-positive ($X$-negative) if $\varphi$ is $X$-positive ($X$-negative, respectively).

For integers $\ell$, $n$, and $r$ such that $\ell\ge 2r+1$, the set of all intervals of consecutive integers $I$ of form
$\{i,i+1,\ldots,i+\ell-1\}$ such that $i\equiv n\pmod{\ell-2r}$ is an \emph{$(\ell,r)$-cover} of integers; note there
are precisely $\ell-2r$ distinct $(\ell,r)$-covers.  For an integer $d\ge 0$, let $M_d(I)$ for the interval $I$ in an $(\ell,r)$-cover
denote the subinterval $\{i+d,\ldots,i+\ell-d-1\}$.
For an integer $m\ge 0$, an \emph{$m$-plan} for a cover $R$
is a function $p:R\to\mathbb{Z}_0^+$ such that $\sum_{I\in R} p(I)=m$ (and in particular, $p$ is non-zero at at most $m$
elements of $R$).

We need the following observation.
\begin{lemma}\label{lemma-cert}
Let $\varphi$ be a basic existential or universal sentence of spread $m$ and range $r'$, over graph language $L$.
Let $C,M\not\in L$ be distinct unary predicate symbols.  Let $G$ be an $L$-interpretation and let $\lambda$ be a layering of $\overline{G}$.
Let $R$ be an $(\ell,r)$-cover of integers for some $r\ge r'$ and $\ell>4r$ and let $p$ be an $m$-plan for $R$.
Suppose that for each $I\in R$ we have
$$G[\lambda^{-1}(I)],C\colonequals\lambda^{-1}(M_{2r}(I)),M\colonequals\lambda^{-1}(M_r(I))\models \varphi^{(p(I))}.$$
Then $G\models \varphi$.
\end{lemma}
\begin{proof}
Let $\psi$ be the core of $\varphi$.  Suppose first that $\varphi$ is a basic existential sentence, and consider any $I\in R$.
Since $G[\lambda^{-1}(I)],C\colonequals\lambda^{-1}(M_{2r}(I))\models \varphi^{(p(I))}$, there exists a set $W_I$ of
$p(I)$ vertices of $\lambda^{-1}(M_{2r}(I))$ such that the distance between them in $G[\lambda^{-1}(I)]$ (and thus also in $G$) is greater than $2r'$
and $G[\lambda^{-1}(I)]\models \psi(w)$ for each $w\in W_I$.  Since $\psi$ is $r'$-local and $r\ge r'$, it also follows that $G\models \psi(w)$ for each $w\in W_I$.
Furthermore, note that for distinct $I,I'\in R$ and $w\in \lambda^{-1}(C_I)$ and $w'\in \lambda^{-1}(C_{I'})$, we have $d_G(w,w')>2r\ge 2r'$.
Hence, the distance between any two vertices in $W=\bigcup_{I\in R} W_I$ is greater than $2r'$, $G\models\psi(w)$ for each $w\in W$, and
$|W|=\sum_{I\in R} p(I)=m$, and thus $G\models \varphi$.

Suppose now $\varphi$ is a basic universal sentence, and for contradiction assume that $G\not\models\varphi$, and thus
there exists a set $W$ of $m+1$ vertices of $\overline{G}$ such that the distance between any two of them is greater than $2r'$
and $G\models\neg\psi(w)$ for each $w\in W$.  Note that $\{M_r(I):I\in R\}$ is a partition of integers, and since $\sum_{I\in R} p(I)=m<|W|$,
there exists $I\in R$ such that $|\lambda^{-1}(M_r(I))\cap W|>p(I)$.  Since $\psi$ is $r'$-local, we have $G[\lambda^{-1}(I)]\models\neg\psi(w)$
for each $w\in \lambda^{-1}(M_r(I))\cap W$.  However, this implies that $G[\lambda^{-1}(I)],M\colonequals\lambda^{-1}(M_r(I))\not\models \varphi^{(p(I))}$,
which is a contradiction.
\end{proof}

A converse holds for a substantial fraction of covers of integers, as we show next.
\begin{lemma}\label{lemma-certex}
Let $\varphi$ be a basic existential or universal sentence of spread $m$ and range $r'$, over graph language $L$.
Let $C,M\not\in L$ be distinct unary predicate symbols.  Let $G$ be an $L$-interpretation such that $G\models \varphi$,
and let $\lambda$ be a layering of $\overline{G}$.  Let $r\ge r'$ and $\ell>2r(3m+1)$ be integers.
For all but at most $6rm$ $(\ell,r)$-covers of integers $R$,
there exists an $m$-plan $p$ for $R$ such that for each $I\in R$ we have
$$G[\lambda^{-1}(I)],C\colonequals\lambda^{-1}(M_{2r}(I)),M\colonequals\lambda^{-1}(M_r(I))\models \varphi^{(p(I))}.$$
\end{lemma}
\begin{proof}
Let $\psi$ be the core of $\varphi$.  Suppose first that $\varphi$ is a basic existential sentence.  Since $G\models \varphi$,
there exists a set $W$ of $m$ vertices of $\overline{G}$ such that the distance between any two vertices of $W$ is greater than $2r'$
and $G\models\psi(w)$ for each $w\in W$.  Note that for each $w\in W$, there are exactly $2r$ $(\ell,r)$-covers of integers $R$ such that
$\lambda(w)$ is not in $\bigcup_{I\in R} M_{2r}(I)$.  Hence, all but at most $2rm$ $(\ell,r)$-covers of integers $R$
satisfy $\lambda(w)\in\bigcup_{I\in R} M_{2r}(I)$ for each $w\in W$.  Consider such an $(\ell,r)$-cover $R$.
Let us define $p(I)=|\lambda^{-1}(M_{2r}(I))\cap W|$ for
each $I\in R$; then $p$ is an $m$-plan for $R$.  Since $\psi$ is $r'$-local and $r\ge r'$, the set $\lambda^{-1}(M_{2r}(I))\cap W$ shows that
$G[\lambda^{-1}(I)],C\colonequals\lambda^{-1}(M_{2r}(I))\models \varphi^{(p(I))}$, as required.

Suppose now that $\varphi$ is a basic universal sentence.  Let $Z$ be the set of vertices $z$ of $\overline{G}$ such that
$G\models\neg\psi(z)$, and let $W$ be a maximal subset of $Z$ such that the distance between any two vertices of $W$ is greater than $2r'$.
Since $G\models \varphi$, we have $|W|\le m$.  Hence, for all but at most $6rm$ $(\ell,r)$-covers of integers $R$
we have $\lambda(w)\in \bigcup_{I\in R} M_{4r}(I)$ for each $w\in W$.  Consider such an $(\ell,r)$-cover $R$.
By the maximality of $W$, every vertex $z\in Z$
is at distance at most $2r'\le 2r$ from $W$, and thus $\lambda(z)\in \bigcup_{I\in R} M_{2r}(I)$.  For $I\in R$, define $p(I)\ge |W\cap \lambda^{-1}(M_{4r}(I))|$
so that $p$ is an $m$-plan.
Note that $Z\cap \lambda^{-1}(M_r(I))$ does not contain more than $p(I)$ vertices at distance more than $2r'$ from one another:
if it contained such a set $W_1$ of size greater than $p(I)$, then note that $W_1\subseteq M_{2r}(I)$, and
thus vertices of $W_1$ are at distance more than $2r\ge 2r'$ from vertices of $W\setminus\lambda^{-1}(M_r(I))$,
and consequently $(W\setminus\lambda^{-1}(M_r(I)))\cup W_1$ would be a subset of $Z$ contradicting the maximality of $W$.
Since $\psi$ is $r'$-local and $r\ge r'$, we have $G[\lambda^{-1}(I)],M\colonequals\lambda^{-1}(M_r(I))\models \varphi^{(p(I))}$, as required.
\end{proof}

For a basic conjunction $\varphi\equiv\bigwedge_{i=1}^t \varphi_i$ and a $t$-tuple $q=(q_1,\ldots,q_t)$ of non-negative integers,
let $\varphi^{(q)}\equiv \bigwedge_{i=1}^t \varphi^{(q_i)}_i$.
Combining the preceding results, we obtain the following key approximation reduction.

\begin{theorem}\label{thm-layap}
Let $\varphi\equiv\bigwedge_{i=1}^t \varphi_i$ be an $X$-positive basic conjunction of spread $m$ and range $r$
over graph language $L\cup\{X\}$, with distinct unary predicate symbols $C,M\not\in L\cup\{X\}$.
For $i=1,\ldots, t$, let $m_i$ be the spread of $\varphi_i$.  Let $\varepsilon_1,\varepsilon_2>0$ be real numbers.
Let $G$ be an $L$-interpretation and let $\lambda$ be a layering of $\overline{G}$.  For any interval $I$ of consecutive integers,
let $G_I$ denote the $(L\cup\{C,M\})$-interpretation $G[\lambda^{-1}(I)],C\colonequals\lambda^{-1}(M_{2r}(I)),M\colonequals\lambda^{-1}(M_r(I))$;
for a $t$-tuple $q$ of non-negative integers, let $A_{I,q}$ be a set such that
$$G_I, X\colonequals A_{I,q}\models \varphi^{(q)}$$
of size at most $(1+\varepsilon_1)\gamma_{\varphi^{(q)}}(G_I)$; $A_{I,q}$ is undefined if $G_I, X\colonequals V(\overline{G_I})\not\models \varphi^{(q)}$.
Let $\ell>2r(1+\max(2/\varepsilon_2,6mt))$ be an integer.  For any $(\ell,r)$-cover $R$ of integers and for any $t$-tuple $p=(p_1,\ldots, p_t)$, where $p_i$ is an $m_i$-plan for $R$
for $i=1,\ldots,t$, let $p(I)=(p_1(I),\ldots,p_t(I))$ and
$$A_{R,p}=\bigcup_{I\in R} A_{I,p(I)};$$
$A_{R,p}$ is undefined if any of the terms on the right-hand side is undefined.
Let $a_{R,p}=\sum_{I\in R} |A_{I,p(I)}|$.
Let $A$ be such a set $A_{R,p}$ with $a_{R,p}$ of minimum over all choices of $R$ and $p$
for which $A_{R,p}$ is defined ($A$ is undefined if no such set $A_{R,p}$ exists).
Then either $G,X\colonequals V(\overline{G})\not\models \varphi$
and $A$ is undefined, or
$$G,X\colonequals A\models \varphi$$
and $|A|\le (1+\varepsilon_1)(1+\varepsilon_2)\gamma_\varphi(G)$.
\end{theorem}
\begin{proof}
Let $R$ and $p$ be such that $A_{R,p}$ is defined, and thus
$G_I, X\colonequals A_{I,p(I)}\models \varphi^{(p(I))}$ for each $I\in R$.
Since $\varphi^{(p(I))}$ is $X$-positive and $A_{I,p(I)}\subseteq A_{R,p}\cap\lambda^{-1}(I)$,
it follows that $G_I,X\colonequals A_{R,p}\cap\lambda^{-1}(I)\models \varphi^{(p(I))}$
for each $I\in R$.  Applying Lemma~\ref{lemma-cert} to $\varphi_1$, \ldots, $\varphi_t$,
we conclude that $G,X\colonequals A_{R,p}\models \varphi$.
In particular, if $A$ is defined, then $G,X\colonequals A\models \varphi$,
and since $\varphi$ is $X$-positive, also $G,X\colonequals V(\overline{G})\models \varphi$.

Suppose now that $G,X\colonequals V(\overline{G})\models \varphi$, and let $A_0$ be a set of size $\gamma=\gamma_\varphi(G)$
such that $G,X\colonequals A_0\models \varphi$.  Let $\RR$ be the set of all $(\ell,r)$-covers $R$ of integers; we have
$|\RR|=\ell-2r$.  For $R\in \RR$, let $\gamma_R=\Bigl(\sum_{I\in R} |A_0\cap \lambda^{-1}(I)|\Bigr)-\gamma$; clearly, $\gamma_R\ge 0$.
Note that
$$\sum_{R\in \RR} \gamma_R=\gamma\ell-\gamma|\RR|=2r\gamma,$$
and thus for at least $|\RR|/2$ elements $R$ of $\RR$ we have
$$\gamma_R\le \frac{4r\gamma}{|\RR|}=\frac{4r}{\ell-2r}\gamma\le \varepsilon_2\gamma.$$
Since $|\RR|/2=\ell/2-r>6rmt$, by Lemma~\ref{lemma-certex}, at least one of these elements $R$ has the following property:
for $i=1,\ldots,t$, there exists a $m_i$-plan $p_i$ for $R$ such that for each $I\in R$ we have
$$G_I,X\colonequals A_0\cap \lambda^{-1}(I)\models \varphi_i^{(p_i(I))},$$
and thus for $p=(p_1,\ldots,p_t)$ we have
$$G_I,X\colonequals A_0\cap \lambda^{-1}(I)\models \varphi^{(p(I))}.$$
In particular, $A_{I,p(I)}$ is defined and
$$|A_{I,p(I)}|\le(1+\varepsilon_1)\gamma_{\varphi^{(p(I))}}(G_I)\le (1+\varepsilon_1)|A_0\cap \lambda^{-1}(I)|.$$
Then
$$a_{R,p}=\sum_{I\in R} |A_{I,p(I)}|\le (1+\varepsilon_1)(\gamma+\gamma_R)\le (1+\varepsilon_1)(1+\varepsilon_2)\gamma,$$
and thus
$|A|\le a_{R,p}\le(1+\varepsilon_1)(1+\varepsilon_2)\gamma_\varphi(G)$ as required.
\end{proof}

We are now ready to prove the main result.

\begin{proof}[Proof of Theorem~\ref{thm-approx}]
Let $\varepsilon_1>\varepsilon_2>\ldots>0$ be an arbitrary sequence such that
$$\prod_{i=1}^\infty (1+\varepsilon_i)\le 1+1/k.$$
Let $\FF_0=\{\varphi\}$.  For $i\ge 1$, we define $\FF_i$, $m_i$, $r_i$, $t_i$ and $\ell_i$ as follows.
Due to Theorem~\ref{thm-local}, we can assume elements of $\FF_{i-1}$ are disjunctions of basic conjunctions.
Let $m_i$ and $r_i$ be the maximum of spreads and ranges of sentences in $\FF_{i-1}$, respectively.
Let $t_i$ be the maximum size of a basic conjunction appearing in one of the disjunctions in $\FF_{i-1}$.
Let $\ell_i=1+\lfloor 2r_i(1+\max(2/\varepsilon_i,6m_it_i))\rfloor$.  Let $\FF_i$ contain, for each $\varphi\in\FF_{i-1}$,
the two corresponding sentences $\varphi_1$ and $\varphi_2$ from Lemma~\ref{lemma-remv}, and for each basic conjunction $\theta$
(of $t_\theta$ terms) appearing in the disjunction $\varphi$ and for each $t_\theta$-tuple $q$ of integers in $\{0,\ldots,m_i\}$
the sentence $\theta^{(q)}$.

Let $\seq{l}=\ell_1,\ell_2,\ldots$; let $t$ be an integer such that Destroyer wins the Baker game on $(G,\seq{l})$ in $t$
rounds for any graph $G\in\CC$ (via the algorithms from the definition of efficiently Baker class).
By Lemma~\ref{lemma-sg}, we can assume $\CC$ is subgraph-closed.  By induction on $i$, we prove the
following claim, from which the theorem follows: For $0\le i\le t$ and any sentence $\theta\in\FF_{t-i}$ over a graph language $L\cup\{X\}$, we can
for an $L$-interpretation $G$ with $n$ vertices such that Destroyer wins (via the aforementioned strategy, given an
appropriate ordering of vertices of $\overline{G}$) the Baker game on $(\overline{G},\tail_{t-i}(\seq{l}))$ in $i$ rounds,
in time $O(n+s(n))$ find a set $A\subseteq V(\overline{G})$ satisfying
$$G,X\colonequals A\models \theta$$
and $|A|\le \Bigl(\prod_{j=t-i+1}^t(1+\varepsilon_j)\Bigr)\gamma_\theta(G)$, or determine no such set $A$ exists. 

Let $\varepsilon'_i=\Bigl(\prod_{j=t-i+1}^t(1+\varepsilon_j)\Bigr)-1$.
If $i=0$, then $V(\overline{G})=\emptyset$ and the claim is trivial; hence, suppose $i\ge 1$.
Consider the first action of Destroyer's strategy on $(\overline{G},\tail_{t-i}(\seq{l}))$,
determined in time $s(n)$.

Suppose first it is a Delete action, moving to $(\overline{G}-v,\tail_{t-i+1}(\seq{l}))$.
Let $\gamma_1$ denote the smallest size of a set $A\subseteq V(\overline{G})$ such that $v\in A$ and
$G,X\colonequals A\models \theta$, and $\gamma_2$ the smallest size of a set $A\subseteq V(\overline{G})$ such that $v\not\in A$ and
$G,X\colonequals A\models \theta$; we have $\gamma_\theta(G)=\min(\gamma_1,\gamma_2)$.
We apply the induction hypothesis to the interpretation $G-v$ with a unary predicate interpreted as $N_G(v)$
and the sentences $\theta_1$ and $\theta_2$ from Lemma~\ref{lemma-remv}, obtaining sets
$A_1,A_2\subseteq V(\overline{G})\setminus \{v\}$ such that
$G,X\colonequals A_1\cup\{v\}\models \theta$, $G,X\colonequals A_2\models \theta$
(or information that no such sets $A_1$ and $A_2$ exist) satisfying
$|A_1|\le (1+\varepsilon'_{i-1})(\gamma_1-1)$ and $|A_2|\le (1+\varepsilon'_{i-1})\gamma_2$.
We let $A$ be the smaller of the sets $A_1\cup\{v\}$ and $A_2$; clearly
$|A|\le (1+\varepsilon'_{i-1})\min(\gamma_1,\gamma_2)=(1+\varepsilon'_{i-1})\gamma_\theta(G)\le (1+\varepsilon'_i)\gamma_\theta(G)$,
as required.

Suppose next Destroyer starts with a Restrict action with layering $\lambda$.
Without loss of generality, we can assume $\theta$ is a basic conjunction (otherwise, we apply the following
procedure to each element of the disjunction and return the largest of the obtained sets) of $t'$ terms
of spreads $m'_1$, \ldots, $m'_{t'}$.  Firstly, by the induction hypothesis we obtain for any interval $I$ of $\ell_{t-i+1}$
consecutive integers such that $\lambda^{-1}(I)\neq\emptyset$ and for any $t'$-tuple $q$ of non-negative integers smaller or equal to $m_{t-i+1}$
a set $A_{I,Q}$ such that $G_I, X\colonequals A_{I,q}\models \theta^{(q)}$ of size at most $(1+\varepsilon'_{i-1})\gamma_{\theta^{(q)}}(G_I)$
(or information that no such set exists), where $G_I$ is as defined in the statement of Theorem~\ref{thm-layap}.

Next, for each $(\ell_{t-i+1},r_{t-i+1})$-cover $R$ of integers, we determine by dynamic programming a $t'$-tuple $p=(p_1,\ldots, p_{t'})$,
where $p_i$ is an $m'_i$-plan for $R$ for $i=1,\ldots,t'$, such that $a_{R,p}$ as defined in the statement of Theorem~\ref{thm-layap}
(with $\theta$ playing the role of $\varphi$) is the smallest possible:  Let $I_1$, \ldots, $I_a$ be the elements of $R$
such that $\lambda^{-1}(I_j)\neq\emptyset$ for $j=1,\ldots,n$, enumerated in any order.  For $j=1,\ldots, a$ and
for any $t'$-tuple $q=(q_1,\ldots,q_{t'})$ of non-negative integers componentwise smaller or equal to $(m'_1,\ldots,m'_{t'})$,
we determine a $t'$-tuple $p_{q,j}$ of $q_1$, \ldots, $q_{t'}$-plans for $R$ with support contained in $\{I_1,\ldots,I_j\}$
such that $a_{R,p_{q,j}}$ is the smallest possible.  Note that to determine $p_{q,j}$, it suffices to take the minimum
of $a_{R,p_{q-l,j-1}}+|A_{I_j,l}|$ over all the possible
$t'$-tuples $l$ of values the plans can assign to $I_j$, which can be done (using the precomputed
values for $j-1$) in time $O(m_{t-i+1}^{t'})$, which is constant.  Hence, this step altogether takes time $O(a)=O(n)$.

Let $A$ be the set $A_{R,p}$ with $a_{R,p}$ minimum, over all choices of $R$ and $p$.  By Theorem~\ref{thm-layap},
such set exists if and only if $G,X\colonequals V(\overline{G})\models \theta$, and in this case
$G,X\colonequals A\models \theta$
$|A|\le (1+\varepsilon_{t-i+1})(1+\varepsilon'_{i-1})\gamma_\theta(G)=(1+\varepsilon'_i)\gamma_\theta(G)$
as required.

For the time complexity, note that in the Restrict case, we recurse for each interval $I$ with only a constant number of sentences,
and thus the total number of vertices of the interpretations is $O(\sum_I |\lambda^{-1}(I)|)=O(\ell_{t-i+1}n)=O(n)$.
Since the recursion has a constant depth, it follows that the total time complexity is $O(n+s(n))$ (plus $f(n)$ for the initialization
of the strategy).
\end{proof}

\subsection{PTAS for negative FO maximization problems}\label{sec-negative}

The algorithm for negative FO maximization problems is analogous to the minimization one; instead of Theorem~\ref{thm-layap},
we use the following reduction.  For an $X$-negative basic conjunction $\varphi$ and a tuple $q$ of integers, let
$\varphi^{(q),C}$ denote the $X$-negative sentence $\varphi^{(q)}\land (\forall x) (C(x)\lor \neg X(x))$.
\begin{theorem}\label{thm-layap-neg}
Let $\varphi\equiv\bigwedge_{i=1}^t \varphi_i$ be an $X$-negative basic conjunction of spread $m$ and range $r$
over graph language $L\cup\{X\}$, with distinct unary predicate symbols $C,M\not\in L\cup\{X\}$.
For $i=1,\ldots, t$, let $m_i$ be the spread of $\varphi_i$.  Let $\varepsilon_1,\varepsilon_2>0$ be real numbers.
Let $G$ be an $L$-interpretation and let $\lambda$ be a layering of $\overline{G}$.  For any interval $I$ of consecutive integers,
let $G_I$ denote the $(L\cup\{C,M\})$-interpretation $G[\lambda^{-1}(I)],C\colonequals\lambda^{-1}(M_{2r}(I)),M\colonequals\lambda^{-1}(M_r(I))$;
for a $t$-tuple $q$ of non-negative integers, let $A_{I,q}$ be a set such that
$$G_I, X\colonequals A_{I,q}\models \varphi^{(q),C}$$
of size at least $(1-\varepsilon_1)\alpha_{\varphi^{(q),C}}(G_I)$; $A_{I,q}$ is undefined if
$G_I, X\colonequals \emptyset\not\models \varphi^{(q),C}$.
Let $\ell>2r(1+\max(2/\varepsilon_2,6mt))$ be an integer.  For any $(\ell,r)$-cover $R$ of integers and
for any $t$-tuple $p=(p_1,\ldots, p_t)$, where $p_i$ is an $m_i$-plan for $R$
for $i=1,\ldots,t$, let $p(I)=(p_1(I),\ldots,p_t(I))$ and
$$A_{R,p}=\bigcup_{I\in R} A_{I,p(I)};$$
$A_{R,p}$ is undefined if any of the terms on the right-hand side is undefined.
Let $A$ be such a set $A_{R,p}$ of maximum size over all choices of $R$ and $p$
for which $A_{R,p}$ is defined ($A$ is undefined if no such set $A_{R,p}$ exists).
Then either $G,X\colonequals \emptyset\not\models \varphi$
and $A$ is undefined, or
$$G,X\colonequals A\models \varphi$$
and $|A|\ge (1-\varepsilon_1)(1-\varepsilon_2)\alpha_\varphi(G)$.
\end{theorem}
\begin{proof}
Let $R$ and $p$ be such that $A_{R,p}$ is defined, and thus
$G_I, X\colonequals A_{I,p(I)}\models \varphi^{(p(I)),C}$ for each $I\in R$.
The term $(\forall x) (C(x)\lor \neg X(x))$ in $\varphi^{(p(I)),C}$
ensures that $A_{R,p}\subseteq \bigcup_{I\in R} \lambda^{-1}(M_{2r}(I))$, and
consequently $A_{I,p(I)}=A_{R,p}\cap\lambda^{-1}(I)$ for each $I\in R$.
We conclude that $G_I,X\colonequals A_{R,p}\cap\lambda^{-1}(I)\models \varphi^{(p(I))}$
for each $I\in R$.  Applying Lemma~\ref{lemma-cert} to $\varphi_1$, \ldots, $\varphi_t$,
we conclude that $G,X\colonequals A_{R,p}\models \varphi$.
In particular, if $A$ is defined, then $G,X\colonequals A\models \varphi$,
and since $\varphi$ is $X$-negative, also $G,X\colonequals \emptyset\models \varphi$.

Suppose now that $G,X\colonequals \emptyset\models \varphi$, and let $A_0$ be a set of size $\alpha=\alpha_\varphi(G)$
such that $G,X\colonequals A_0\models \varphi$.  Let $\RR$ be the set of all $(\ell,r)$-covers $R$ of integers; we have
$|\RR|=\ell-2r$.  For $R\in \RR$, let $\alpha_R=\alpha-\Bigl(\sum_{I\in R} |A_0\cap \lambda^{-1}(M_{2r}(I))|\Bigr)$; clearly, $\alpha_R\ge 0$.
Note that
$$\sum_{R\in \RR} \alpha_R=\alpha|\RR|-\alpha(\ell-4r)=2r\alpha,$$
and thus for at least $|\RR|/2$ elements $R$ of $\RR$ we have
$$\alpha_R\le \frac{4r\alpha}{|\RR|}=\frac{4r}{\ell-2r}\alpha\le \varepsilon_2\alpha.$$
By Lemma~\ref{lemma-certex}, at least one of these elements $R$ has the following property:
for $i=1,\ldots,t$, there exists a $m_i$-plan $p_i$ for $R$ such that for each $I\in R$ we have
$$G_I,X\colonequals A_0\cap \lambda^{-1}(I)\models \varphi_i^{(p_i(I))},$$
and thus for $p=(p_1,\ldots,p_t)$ we have
$$G_I,X\colonequals A_0\cap \lambda^{-1}(I)\models \varphi^{(p(I))}.$$
Since $\varphi^{(p(I))}$ is $X$-negative, it follows that
$$G_I,X\colonequals A_0\cap \lambda^{-1}(M_{2r}(I))\models \varphi^{(p(I))},$$
and thus
$$G_I,X\colonequals A_0\cap \lambda^{-1}(M_{2r}(I))\models \varphi^{(p(I)),C}.$$
In particular, $A_{I,p(I)}$ is defined and
$$|A_{I,p(I)}|\ge(1-\varepsilon_1)\alpha_{\varphi^{(p(I),C)}}(G_I)\ge (1-\varepsilon_1)|A_0\cap \lambda^{-1}(M_{2r}(I))|.$$
Then
\begin{align*}
|A_{R,p}|&=\sum_{I\in R} |A_{I,p(I)}|\ge (1-\varepsilon_1)\sum_{I\in R} |A_0\cap \lambda^{-1}(M_{2r}(I))|\\
&=(1-\varepsilon_1)(\alpha-\alpha_R)\ge (1-\varepsilon_1)(1-\varepsilon_2)\alpha,
\end{align*}
and thus
$|A|\ge |A_{R,p}|\ge (1-\varepsilon_1)(1-\varepsilon_2)\alpha_\varphi(G)$ as required.
\end{proof}

A straightforward variation on the proof of Theorem~\ref{thm-approx} now gives Theorem~\ref{thm-approx-max}.

\subsection{Non-FO problems}\label{sec-other}

While we presented the PTASes in the meta-algorithmic setting of monotone first-order optimization problems,
the first-order definability is not crucially necessary for the technique to apply; we only use it to derive the
required locality properties.  In particular, most problems amenable to the Baker's technique can be dealt with;
one only needs to determine how to handle the vertex deletion, which usually does not cause difficulties (although it
may make it necessary to work with a more general problem to keep track of the information about the deleted vertices).
Let us mention that most of previous works on the topic eventually reduce the problem to the bounded treewidth case;
we end with an empty graph, and thus we avoid requiring tractability of the problem on bounded treewidth graphs
(but anyway, this generally is not an issue).

As an example, let us consider the problem of finding a largest $c$-colorable induced subgraph (for a fixed constant $c$);
this is a monotone maximization problem.  We will need to consider a more general problem: Given a graph $G$ and an assignment $L$
of lists of size at most $c$ to vertices of $G$, find a largest induced $L$-colorable subgraph.  When dealing with a Delete action
removing a vertex $v$, we take the largest of the outcomes from the following subcases:
\begin{itemize}
\item For each color $a\in L(v)$,
recurse on $G-v$ with the list assignment obtained from $L$ by removing $a$ from the lists of neighbors of $v$, and add $v$ to the
subgraph obtained from this recursive call.
\item Recurse on $G-v$ with the same list assignment $L$.
\end{itemize}
When dealing with the Restrict action with layering $\lambda$, we take the largest of the outcomes of the following
procedure applied to all $(\ell,0)$-covers $R$ of integers:
For each $I\in R$, we recurse on $G[M_1(I)]$ (with the same list assignment) and take the union of all returned subgraphs
(they are non-adjacent, so their colorings do not conflict).  As there is a choice of $R$ where $\bigcup_{I\in R} M_1(I)$ covers
at least $(1-2/\ell)$ fraction of an optimal solution, this is guaranteed to multiply the approximation ratio by at most
$(1-2/\ell)$.  Hence, when we run the game over a sufficiently fastly growing sequence of $\ell$'s,
the total approximation ratio can be arbitrarily close to $1$.

\section*{Acknowledgments}

I would like to thank Sebastian Siebertz for helpful discussions.

\bibliographystyle{acm}
\bibliography{thin}

\end{document}